\documentclass[10pt,journal]{IEEEtran}
\usepackage{cite}
\usepackage{amsmath,amssymb,amsthm}
\usepackage{graphicx}
\usepackage[caption=false,font=footnotesize]{subfig}
\usepackage{booktabs,array,multirow}
\usepackage{algorithmic}
\usepackage{url}
\usepackage{microtype}
\usepackage[hidelinks]{hyperref}
\graphicspath{{figures/}}

\newtheorem{theorem}{Theorem}
\newcommand{\rtm}{\ensuremath{(R,T,M)}}
\newcommand{\refcost}{\ensuremath{(R_{\rm ref},T_{\rm ref},M_{\rm ref})}}
\newcommand{\DPQA}{OLSQ-DPQA}
\newcommand{\figwide}[3]{\begin{figure*}[t]\centering\includegraphics[width=\textwidth]{#1}\caption{#2}\label{#3}\end{figure*}}
\newcommand{\figsingle}[3]{\begin{figure}[t]\centering\includegraphics[width=\columnwidth]{#1}\caption{#2}\label{#3}\end{figure}}

\begin{document}
\raggedbottom
\title{NAKO: A Benchmarking Framework for Evaluating and Diagnosing Neutral-Atom Quantum Compilers}
\author{Wenjie~Sun, Zhigang~Wang, Xiaoyu~Li, Desheng~Zheng, Lianhui~Yu, Geng~Chen, and Guowu~Yang%
\thanks{(Corresponding author: Zhigang Wang.)}%
\thanks{Wenjie Sun and Zhigang Wang are with the School of Electronic Science and Engineering, University of Electronic Science and Technology of China, No.~2006, Xiyuan Ave., Chengdu, Sichuan 611731, China (e-mail: 202211022528@std.uestc.edu.cn; zhigangwang@uestc.edu.cn).}%
\thanks{Xiaoyu Li is with the School of Information and Software Engineering, University of Electronic Science and Technology of China, No.~2006, Xiyuan Ave., Chengdu, Sichuan 611731, China (e-mail: xiaoyuuestc@uestc.edu.cn).}%
\thanks{Desheng Zheng is with the School of Computer Science and Software Engineering, Southwest Petroleum University, No.~8, Xindu Road, Chengdu, Sichuan 610500, China (e-mail: zheng\_de\_sheng@163.com).}%
\thanks{Lianhui Yu, Geng Chen, and Guowu Yang are with the School of Computer Science and Engineering, University of Electronic Science and Technology of China, No.~2006, Xiyuan Ave., Chengdu, Sichuan 611731, China (e-mail: 202222120303@std.uestc.edu.cn; 202312081614@std.uestc.edu.cn; guowu@uestc.edu.cn).}}
\maketitle

\begin{abstract}
Reconfigurable neutral-atom processors couple gate scheduling, atom placement, and parallel atom transport, so compiler decisions jointly determine physical execution costs. Evaluating these compilers requires attainable resource references and a way to trace excess costs to their causes. We present NAKO, a benchmarking framework that parameterizes circuits by qubit count, logical depth, and requested one- and two-qubit gate densities, and assigns attainable targets for Rydberg stages, routing intervals, and movement layers. Each benchmark is paired with an executable reference schedule; under the stated execution model, construction guarantees minimum Rydberg-stage and routing-interval counts, and minimum movement among schedules with the minimum Rydberg-stage count. Comparing compiler outputs with these references quantifies resource gaps, while reference mappings and intermediate physical states enable controlled interventions that trace excess costs to specific compiler decisions. In a 648-run evaluation of Enola and OLSQ-DPQA, NAKO identifies distinct bottlenecks. For Enola, placement scoring and cross-stage layout management contribute to excess movement; diagnosis-guided interventions reduce mean movement by 78.9\% across five matched circuits. For OLSQ-DPQA, premature candidate selection and limited exploration of preparatory states produce avoidable entangling stages and routing intervals; four fixed-checkpoint continuations each eliminate one Rydberg stage and one routing interval without increasing movement. NAKO thus links resource targets backed by constructive optimality guarantees to controlled diagnosis and validation of compiler-specific resource savings.
\end{abstract}


\section{Introduction}
\label{sec:introduction}
\IEEEPARstart{R}{econfigurable} neutral-atom processors combine scalable atomic arrays with connectivity that can change during execution. Progress in defect-free array assembly and programmable many-body systems has established the physical basis for arranging and controlling large registers~\cite{barredo2016,browaeys2020,ebadi2021}. Parallel Rydberg gates and multi-qubit experiments have advanced their computational capabilities~\cite{levine2019,graham2019,graham2022,evered2023}, while coherent transport enables entangled qubits to move between interaction configurations~\cite{bluvstein2022}. Logical processors and fault-tolerant architectural proposals further increase the importance of coordinating operations across these configurations~\cite{bluvstein2024,xu2024}. Translating this flexibility into efficient execution depends on the compiler: unnecessary entangling pulses and atom transport consume time and coherence that could support useful computation.

Quantum compilation has long addressed the gap between logical interactions and physical connectivity. Placement and allocation methods establish logical-to-physical assignments~\cite{maslov2008,siraichi2018}; heuristic routing, exact mapping, and scalable layout synthesis optimize the additional operations required by restricted connections~\cite{li2019sabre,wille2019,tan2020,lin2023}. Execution-time optimization and noise-adaptive mapping show that the appropriate objective also depends on operation durations and device characteristics~\cite{zhang2021,murali2019noise}. Full-stack comparisons demonstrate the resulting interaction between architectural properties and compiler choices~\cite{murali2019full}. For neutral atoms, long-range interactions, topology displacement, and native multi-qubit operations introduce additional opportunities~\cite{baker2021,brandhofer2021,patel2022}. Recent compilers coordinate dynamic placement, parallel movement, and gate scheduling through constraint solving, decomposed optimization, and mobile ancillas~\cite{tan2024olsq,wang2024atomique,wang2024qpilot,tan2025enola}. Evaluating these methods requires benchmarks that expose the physical costs of their decisions.

The first challenge is to establish an attainable performance reference for each workload. Application benchmarks provide important coverage: SupermarQ supplies scalable workloads and structural features~\cite{tomesh2022}; QASMBench collects low-level circuits and characterization metrics~\cite{li2023qasm}; and MQT Bench supplies consistent workloads across compilation abstraction levels~\cite{quetschlich2023}. On such workloads, observed differences show which compiler performs better under the tested conditions. Quantifying the remaining improvement opportunity additionally requires knowledge of what execution costs the circuit can attain. This need becomes particularly acute for neutral atoms, where fewer entangling stages can coexist with more transport. A reference must capture multiple physical resources and demonstrate that their combination is realizable.

The second challenge is to diagnose the compiler decisions behind an observed resource gap. Excess movement may arise from initial placement, loss of reusable configurations, or overly fragmented routing batches, while additional entangling stages may result from limited parallelism or search that misses useful preparatory motion. Aggregate costs reveal the symptom but not its cause; controlled comparisons under the same circuit conditions are needed to isolate the responsible decision and determine whether improvement requires a better local objective or a richer search space.

Constructive benchmarks address the reference problem by embedding solutions with known properties. QUEKO provides circuits with known optimal depth and gate counts and mappings requiring no inserted SWAPs~\cite{tan2021}; QKNOB embeds low-overhead transformations to evaluate SWAP and depth overhead~\cite{li2025qknob}; and QUBIKOS provides provably optimal nonzero SWAP costs and analyzes suboptimal routing decisions~\cite{ping2025}. These frameworks target fixed coupling graphs, as found in superconducting processors. For neutral atoms, Khusainov et al. unify compiler-output representations and evaluation models for cross-toolchain comparison~\cite{khusainov2026}, while Gao et al. develop linear and zigzag compilation strategies for the quantum Fourier transform (QFT), with analytical bounds on large atom movements~\cite{gao2026qft}.

Taken together, these efforts leave a need for a general framework that connects resource evaluation with compiler diagnosis under reconfigurable neutral-atom execution. Gate scheduling, atom placement, and parallel transport jointly determine physical execution costs, so a benchmark must provide not only an executable reference with jointly attainable targets for Rydberg stages, routing intervals, and movement layers, but also physical states that enable controlled investigation of the decisions behind resource gaps. Without these two capabilities, compiler outputs can be compared or individual resource bounds can be established, but systematically quantifying avoidable overhead across coupled resources and tracing it to specific compiler decisions remains difficult.

We present NAKO, a benchmarking framework designed to address this need. NAKO constructs circuits parameterized by qubit count, logical depth, and requested one- and two-qubit gate densities, together with a target resource tuple $(R,T,M)$. Each circuit is accompanied by an executable physical reference schedule that attains this tuple, with constructive guarantees of minimum Rydberg-stage and routing-interval counts and minimum movement among schedules with the minimum Rydberg-stage count. NAKO also provides a reference-guided workflow comprising resource comparison, controlled intervention, and execution validation. Its reference layouts support initial-mapping comparisons for identical circuits, while alternative gate and movement sequences are examined from the tested compiler's own intermediate state. Figure~\ref{fig:framework} summarizes the framework, and Section~\ref{sec:method} develops its construction, resource guarantees, and diagnostic workflow.

\figsingle{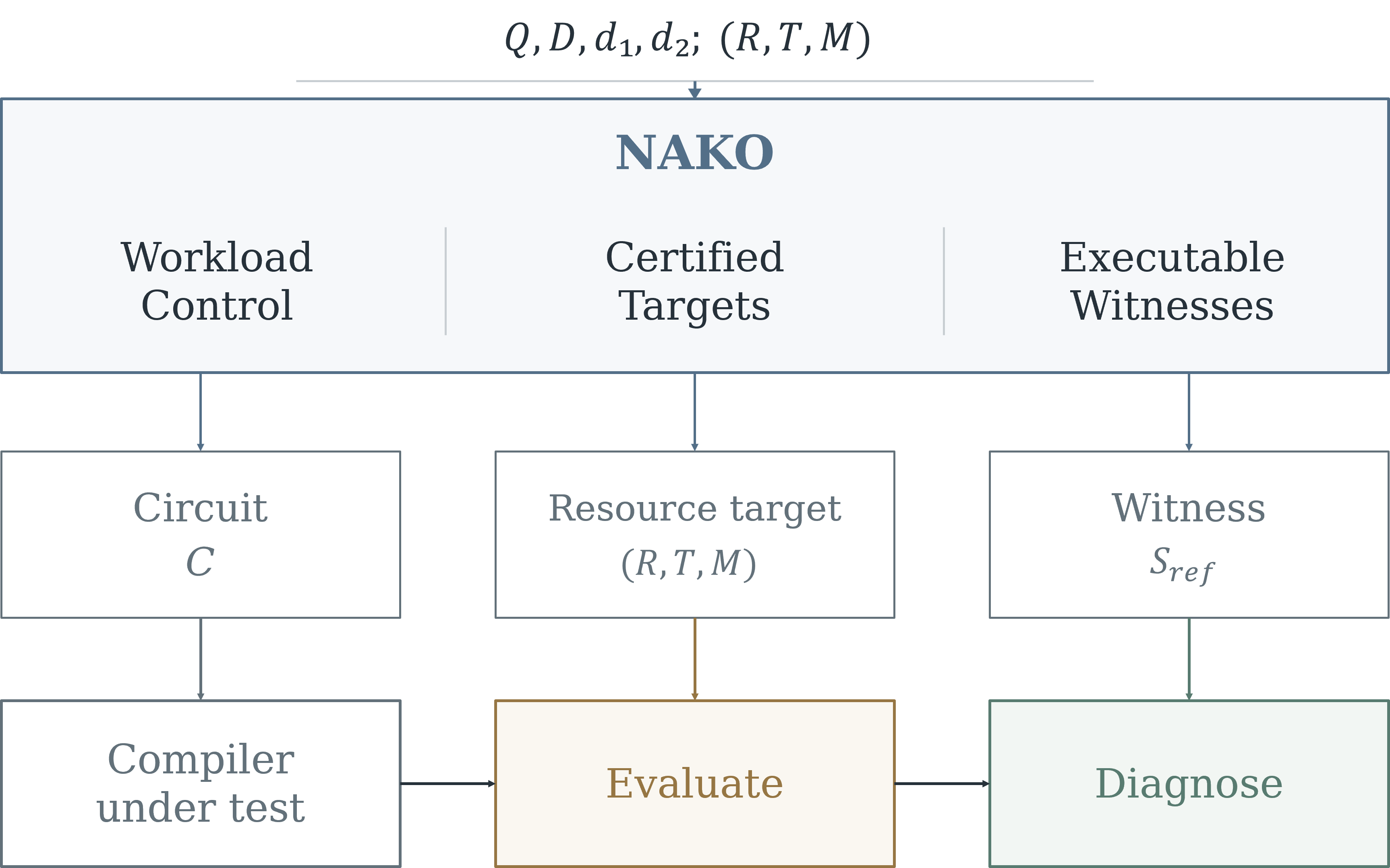}{Overview of NAKO. Given workload and resource parameters, NAKO constructs a benchmark circuit $C$, a certified resource target $(R,T,M)$, and an executable physical witness $S_{\rm ref}$. The compiler under test compiles $C$; NAKO compares the resulting schedule with the certified target to quantify resource gaps, and uses $S_{\rm ref}$ to diagnose the compiler decisions responsible for those gaps.}{fig:framework}

The contributions are as follows.
\begin{itemize}
\item \textbf{Parameterized circuit construction.} NAKO generates circuits with configurable qubit counts, depths, and gate densities, together with prescribed Rydberg-stage, routing-interval, and movement-layer targets.
\item \textbf{Constructive optimality guarantees.} An eight-qubit backbone and three composable primitives yield executable schedules and proofs of minimum Rydberg and routing counts, with minimum movement among minimum-Rydberg schedules.
\item \textbf{Scalable generation.} Direct construction produces 2000-qubit benchmarks in 4.97~s on average, enabling large-scale compiler testing.
\item \textbf{Demonstrated compiler improvements.} Evaluation across 648 runs guides targeted interventions in Enola and Optimal Layout Synthesis for Quantum Computing on dynamically field-programmable qubit arrays (OLSQ-DPQA), including a 78.9\% reduction in Enola's mean movement over five matched circuits.
\end{itemize}

\section{Background and Resource Metrics}
\label{sec:architecture}
As discussed in Section~\ref{sec:introduction}, comparing compiler outputs alone leaves unclear how much physical resource overhead is avoidable and which decisions cause it. For neutral atoms, resolving these questions requires accounting for both gate execution and the atom arrangements and movements that enable it. This section introduces the compilation process, its physical constraints, and the resource metrics used by NAKO. Figure~\ref{fig:architecture} connects the constraints to the execution events counted by these metrics.

\figwide{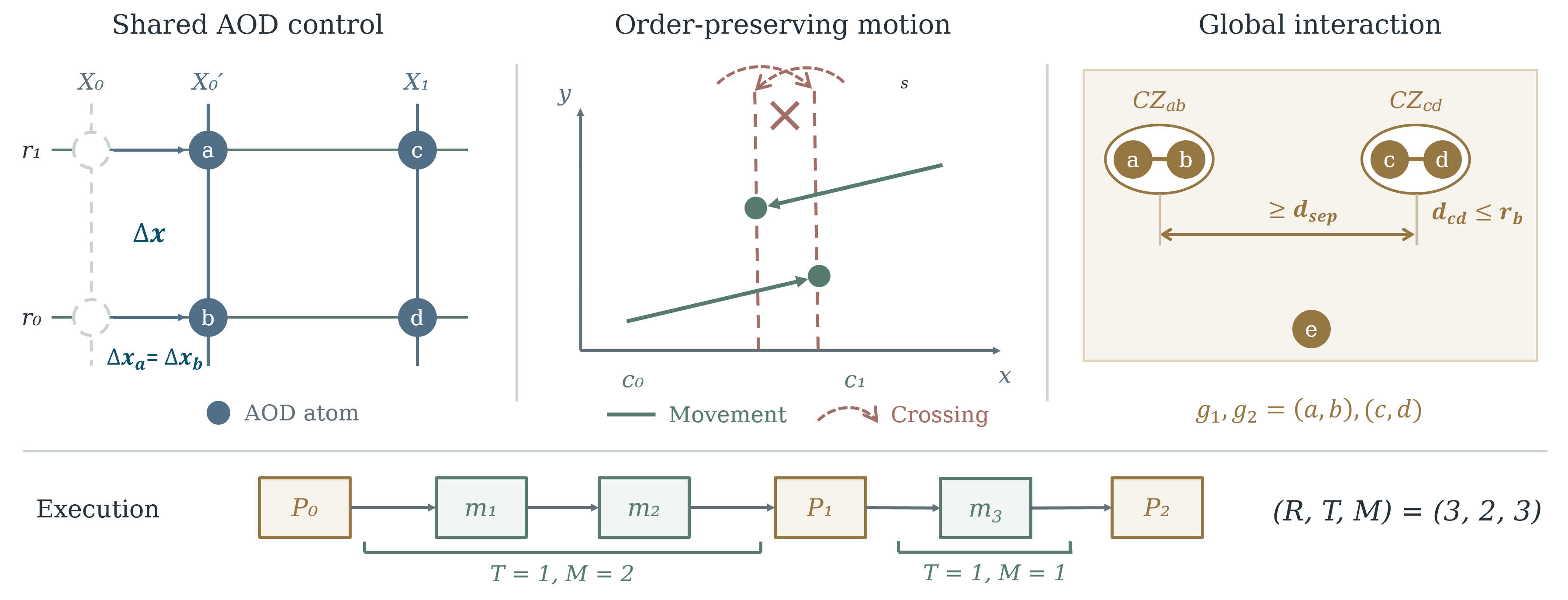}{Physical constraints and resource accounting for reconfigurable neutral-atom execution. Left: atoms sharing an AOD column move together; dashed circles mark their previous positions. Center: the illustrated movements reverse column order and conflict within one AOD array. Right: a global pulse executes CZ gates on the isolated pairs $(a,b)$ and $(c,d)$. Bottom: two routing intervals contain two and one movement layers, respectively, giving $(R,T,M)=(3,2,3)$. Geometry and event spacing are schematic. The constraints follow~\cite{tan2024olsq,tan2025enola,wang2024atomique,wang2024qpilot}.}{fig:architecture}

\subsection{From Logical Gates to Physical Execution}
The input circuit specifies operations on $Q$ qubits and their dependencies. Its logical depth $D$ counts dependency layers, while $d_1$ and $d_2$ describe the requested single- and two-qubit gate densities. A gate becomes ready when its predecessors have executed; disjoint ready gates are candidates for parallel execution. Throughout our compiler experiments, gates sharing a qubit retain their input order.

Executing these candidates on neutral atoms requires an additional spatial decision. Each qubit is carried by an atom held in an optical trap. A spatial light modulator (SLM) supplies stationary traps, while acousto-optic deflectors (AODs) supply movable traps. AOD motion transports trapped atoms while preserving their quantum states. Transfers between aligned static and mobile traps change which atoms are carried by the AOD~\cite{bluvstein2022,tan2024olsq,tan2025enola}. The compiler therefore selects a gate batch, arranges its atoms, and routes between successive arrangements. These choices produce a physical schedule containing gates, trap transfers, and movements. The constraints below determine which such schedules are feasible.

\subsection{Constraints on Physical Execution}
\textit{Shared motion and preserved order.}
A crossed AOD forms traps at intersections of controlled rows and columns. An atom $q$ in row $r(q)$ and column $c(q)$ has coordinates
\begin{equation}
x_q(t)=X_{c(q)}(t),\qquad y_q(t)=Y_{r(q)}(t).
\end{equation}
All atoms carried by a column share its horizontal displacement; those carried by a row share its vertical displacement. In the left view of Fig.~\ref{fig:architecture}, shifting one column moves both $a$ and $b$. The compiler must therefore group compatible motions and use transfers when a different grouping is needed.

For atoms carried by the same AOD array in one movement layer, shared control and order preservation require~\cite{tan2024olsq,tan2025enola,wang2024atomique,wang2024qpilot}
\begin{equation}
\operatorname{sgn}(u_0^{s}-u_1^{s})
=\operatorname{sgn}(u_0^{e}-u_1^{e}),\quad u\in\{x,y\},
\end{equation}
where superscripts $s$ and $e$ denote the start and end coordinates. Thus, atoms sharing a row or column remain aligned, and distinct rows and columns retain their relative order. The center view of Fig.~\ref{fig:architecture} illustrates a conflict: the proposed destinations reverse the atoms' horizontal order. Distinct rows and columns also maintain a minimum physical separation $d_s$ throughout motion. Discretized site coordinates may group nearby traps into the same site, while their physical traps remain distinct and ordered. Occupancy and alignment constraints govern loading and unloading.

\textit{Geometry-dependent, global interactions.}
Once the atoms are arranged, a Rydberg pulse realizes the selected two-qubit operations. In the global-pulse model considered here, an isolated pair within the Rydberg interaction range $r_b$ undergoes a controlled-$Z$ (CZ) gate. Other atoms must be sufficiently separated to suppress unwanted interactions; denote this separation by $d_{\rm sep}$. The cited compilation models use a separation on the order of $2.5r_b$~\cite{tan2024olsq,tan2025enola,wang2024qpilot}. Single-qubit gates can be individually addressed between Rydberg stages~\cite{tan2024olsq,tan2025enola}.

The pulse illuminates the entire modeled interaction region and executes a CZ on every isolated close pair. In the right view of Fig.~\ref{fig:architecture}, the pairs $(a,b)$ and $(c,d)$ lie within interaction range and are sufficiently separated from other atoms, allowing both gates to execute in one pulse. The isolated atom $e$ is also illuminated and participates in no two-qubit gate. A valid pulse requires the scheduled gate pairs to match the physical interaction pairs throughout the illuminated region. This couples gate scheduling to the complete atom placement.

We represent this geometry by separated interaction sites, each with two distinct atom slots. A doubly occupied site realizes a pair; different sites separate independent pairs and isolated atoms. This abstraction follows the site discretization used in the compiler literature~\cite{tan2024olsq,tan2025enola}. Together, the motion and interaction constraints explain why a logically parallel batch can require several transport steps before one pulse can execute it.

Geometry also determines when a layout can be reused: successive pulses at an unchanged configuration repeat its entire pair matching. Reuse is feasible when the remaining program calls for those gates; a different matching requires the interaction geometry to change.

\subsection{From Physical Execution to Resource Cost}
The same constraints create distinct sources of overhead. Additional pulses expose atoms to further gate and excitation errors; routing and trap transfers consume execution time and contribute to the physical error budget~\cite{tan2025enola}. A single logical-depth measure therefore leaves important execution costs unresolved. For a physical schedule $S$, we record:
\begin{itemize}
\item $R(S)$: the number of nonempty Rydberg entangling pulses.
\item $T(S)$: the number of intervals between consecutive nonempty pulses that contain at least one nonzero movement.
\item $M(S)$: the total number of nonzero movement layers, including movements before the first or after the last pulse.
\end{itemize}
A movement layer can transport multiple atoms simultaneously. In Fig.~\ref{fig:architecture}, the $P_0$--$P_1$ interval contains $m_1,m_2$, and the $P_1$--$P_2$ interval contains $m_3$. Their contributions are $(\Delta T,\Delta M)=(1,2)$ and $(1,1)$, respectively, yielding $(R,T,M)=(3,2,3)$. Thus, $T$ measures the frequency of routing between entangling stages, and $M$ counts the movement layers accumulated across the schedule. Trap transfers remain explicit physical operations in the execution model.

\begin{table}[t]
\caption{Notation and measurement quantities.}
\label{tab:notation}\centering\small
\begin{tabular}{@{}ll@{}}\toprule
Symbol & Meaning\\\midrule
$Q,D$ & Logical width and circuit depth\\
$d_1,d_2$ & Requested single- and two-qubit gate densities\\
$R,T,M$ & Rydberg stages, routing intervals, movements\\
$n_{100},n_{111},n_{223}$ & Appended primitive counts\\
$S_{\rm ref}$ & Constructed physical reference schedule\\
$\rho_X=X_c/X_{\rm ref}$ & Compiler-to-reference resource ratio\\
$\Delta X=X_c-X_{\rm ref}$ & Absolute resource excess\\\bottomrule
\end{tabular}
\end{table}

We compare each compiler output with the reference for the same input through $\rho_X$ and $\Delta X$ for $X\in\{R,T,M\}$ (Table~\ref{tab:notation}). Enola exposes movement through nonzero moves in its expanded instruction stream, including transport associated with loading and unloading. OLSQ-DPQA exposes movement through coordinate-changing transitions between encoded states. The interpretation retains these implementation-specific granularities. Paired interventions use the same compiler's constraints and accounting; duration and estimated fidelity provide additional comparisons under its unchanged physical model.

\subsection{Motivation for Benchmark Construction}
Gate scheduling, atom placement, and transport jointly determine the physical costs defined above. The resulting $R$, $T$, and $M$ values quantify resource use, but comparing these values across compilers leaves two questions unresolved: how much overhead is avoidable, and which decisions produce it. Answering the first requires resource targets that a valid schedule can attain under the coupled motion and global-pulse constraints. Answering the second requires controlled comparisons that connect the measured gap to specific compiler decisions.

These requirements motivate constructing benchmark circuits together with their physical reference schedules. NAKO follows this approach to combine controllable workloads with attainable resource targets. Given a qubit count, logical depth, requested single- and two-qubit gate densities, and a feasible target tuple \rtm, NAKO produces a circuit in Quantum Assembly Language (QASM) and an executable schedule attaining that tuple. Under the dependency-preserving execution model considered here, $R$ and $T$ are minimal, and $M$ is minimal among executions with minimum $R$. Direct construction also supports fast generation for parameter sweeps and large compiler tests.

By pairing each circuit with known attainable resource targets, NAKO provides a concrete reference for quantifying resource gaps in compiler outputs. Its executable reference schedules also supply layouts and execution information for controlled tests of compiler decisions. Through this framework, we aim to identify avoidable overhead and guide targeted compiler improvements.

\section{NAKO Benchmarking Framework}
\label{sec:method}
\subsection{Overview}
NAKO connects three stages: construction supplies a circuit and its physical reference; evaluation measures resource gaps and sensitivity; diagnosis tests the compilation decisions suggested by those measurements. The compiler under test receives the logical circuit and its hardware configuration. Resource targets and witness annotations remain available to the evaluation procedure, with initial mappings supplied explicitly in the paired mapping study.

Figure~\ref{fig:generation} shows how NAKO constructs a circuit with the prescribed resources. It first builds an eight-qubit core, called the backbone, whose ordered CZ gates establish the resource bounds under the stated execution conditions. Small reusable construction blocks, called primitives, extend this core with known increments in \rtm; the requested tuple determines how many instances of each primitive are needed. Additional gates, called filler gates, adjust circuit depth and gate load while preserving the resource targets. The required atom movements are recorded during construction, so the logical circuit and its physical reference schedule are produced together.

\figwide{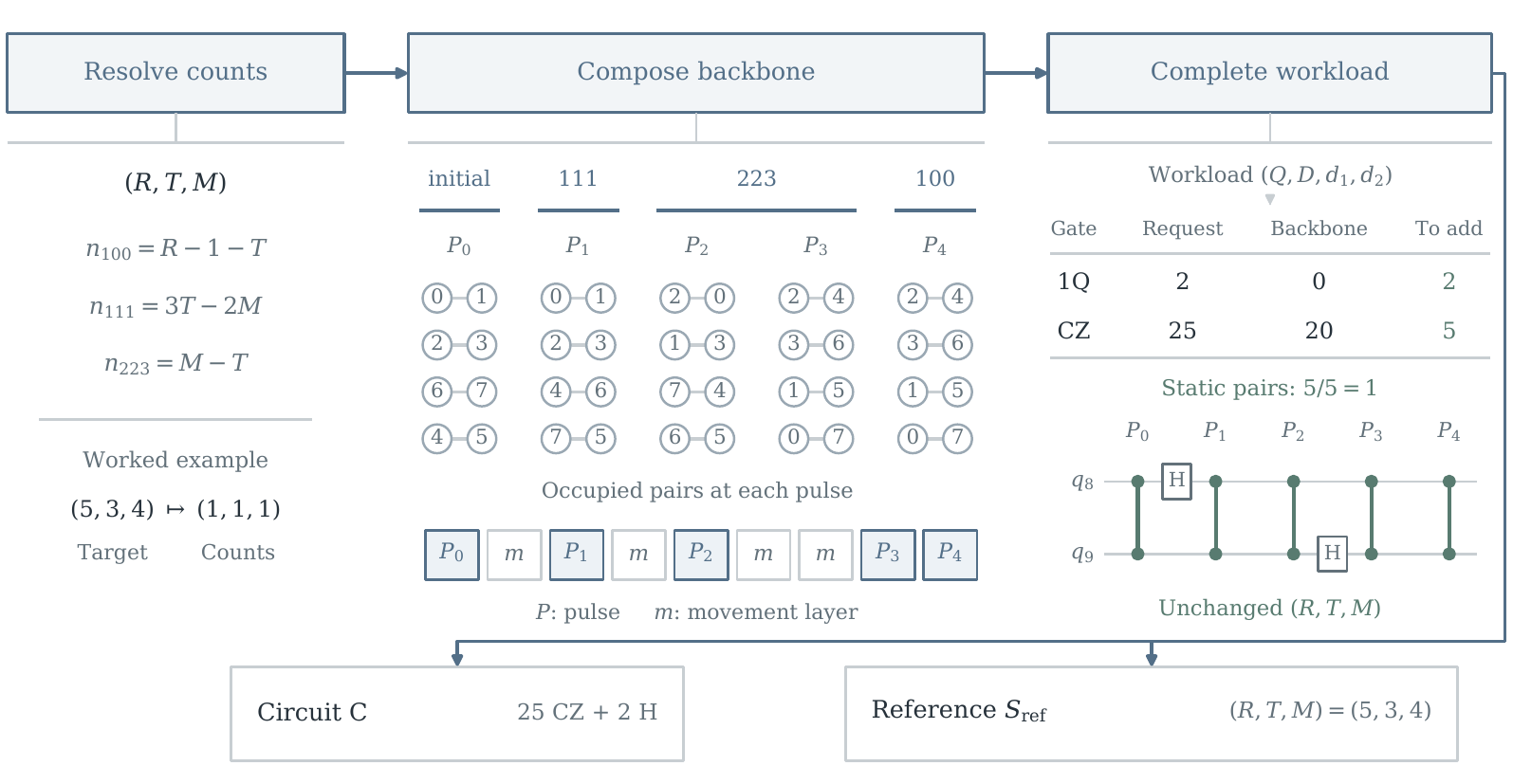}{Construction with target $(R,T,M)=(5,3,4)$. Initial matching, $111$, $223$, and $100$ produce five pulses with movement counts $1,1,2,0$ between them. Workload parameters determine the requested gate counts. In this example, subtracting the backbone's 20 CZ gates from the request of 25 leaves five CZ gates to add; one static pair supplies these across the five pulses. The single-qubit request is completed with two Hadamard (H) gates. The resulting circuit contains 25 CZ gates and two H gates, and its reference retains $(R,T,M)=(5,3,4)$. Numbered pairs show backbone states; event widths indicate order.}{fig:generation}

\subsection{Backbone and Primitive Interfaces}
The backbone provides the gate dependencies used to establish the resource bounds. It uses eight qubits, each participating in one CZ gate per backbone stage. Initially, these qubits form four disjoint pairs, with matching
\begin{equation}
P_0=\{(0,1),(2,3),(4,5),(6,7)\}.
\end{equation}
Every subsequent backbone stage is also a perfect matching on these qubits: four disjoint CZ gates cover all eight qubits. Each backbone qubit therefore follows a chain of $R$ ordered CZ gates. Preserving this order requires at least $R$ nonempty entangling stages, and the reference schedule attains this bound.

To connect successive primitives, the constructor must specify the physical state from which each primitive starts. We call this state its interface: the occupied pair sites and the qubit-to-slot assignments within them. Each primitive supplies the gates and movements leading to its output interface. When a template is reused, its qubit roles are assigned to the qubits currently occupying the input slots.

The first row reserves eight columns as movement workspace for the primitives. These eight sites contain the eight backbone atoms; other atoms occupy sites outside this region. With two slots per site, the remaining sites accommodate $2(rc-8)$ atoms, giving a total capacity of $8+2(rc-8)=2rc-8$ qubits.

\begin{figure*}[t]
\centering
\subfloat[$100$: retained matching.]{\includegraphics[width=.269\textwidth]{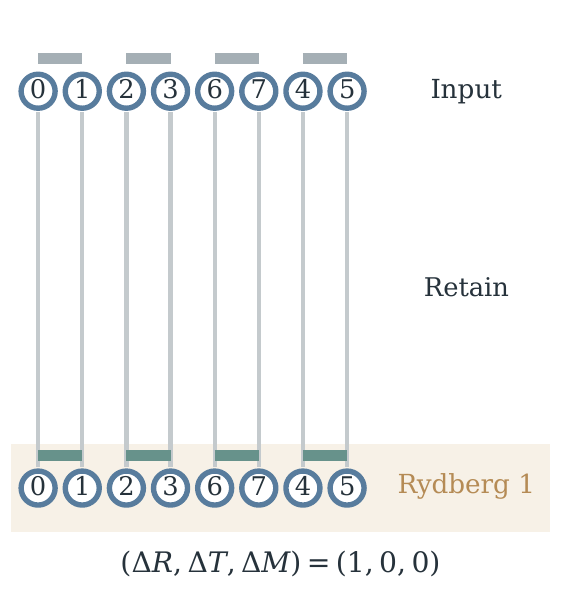}}\hfill
\subfloat[$111$: one movement.]{\includegraphics[width=.297\textwidth]{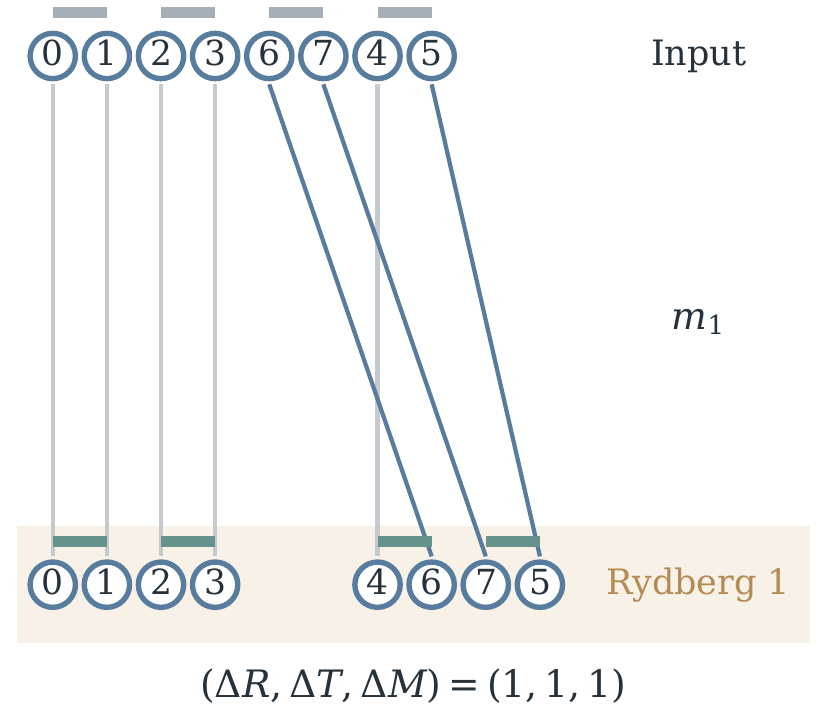}}\hfill
\subfloat[$223$: three movements.]{\includegraphics[width=.410\textwidth]{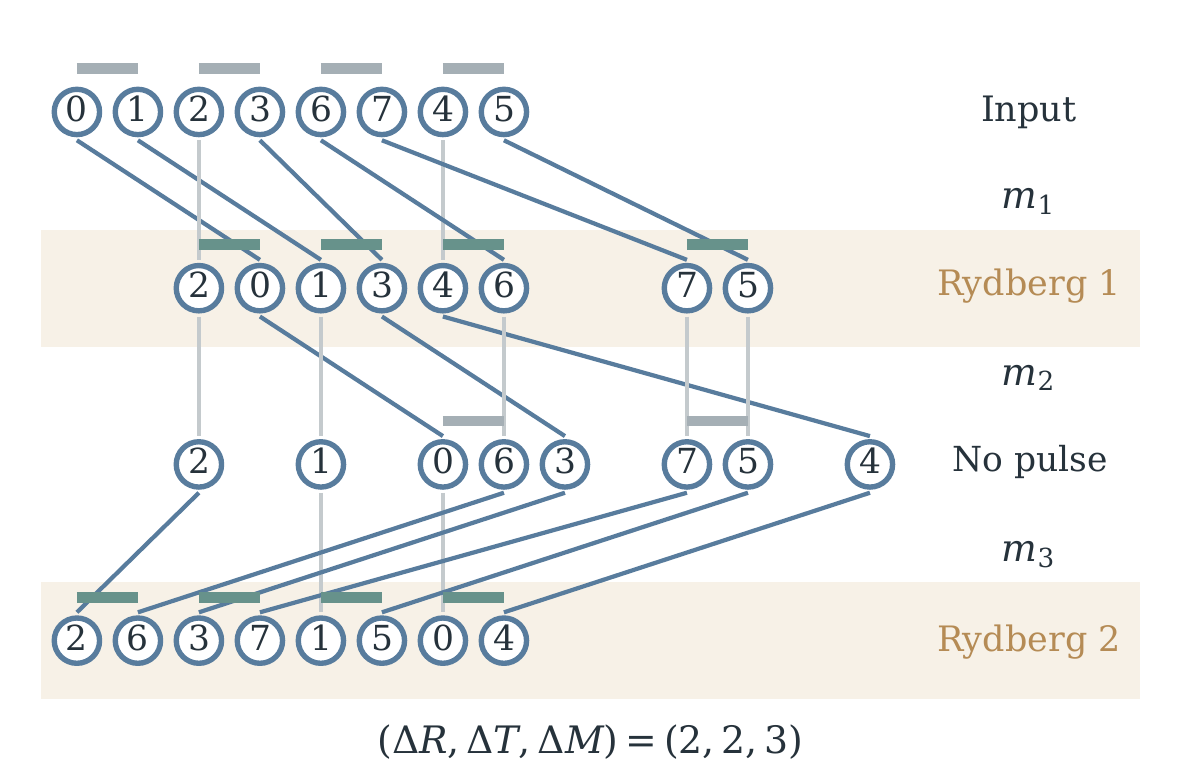}}
\caption{Primitive execution witnesses with time advancing downward. Numbered circles identify backbone qubits, and connecting trajectories show coordinate changes. Input denotes the inherited configuration. Shaded rows labeled Rydberg mark newly executed entangling stages; green pair marks indicate their CZ gates. Gray pair marks show co-location in the input and intermediate configurations. The $223$ primitive executes Rydberg stages after $m_1$ and $m_3$; its intermediate configuration after $m_2$ has no pulse. The $100$ and $111$ primitives each add one Rydberg stage.}
\label{fig:primitives}
\end{figure*}

\subsection{Three Composable Primitives}
The three primitives in Fig.~\ref{fig:primitives} extend the backbone with different resource increments. Their names encode $(\Delta R,\Delta T,\Delta M)$ relative to the input state. They respectively repeat a pairing without movement, change partners in one movement, and realize a two-interval segment requiring three movements. Their resource vectors are
\begin{equation}
v_{100}=(1,0,0),\quad v_{111}=(1,1,1),\quad v_{223}=(2,2,3).
\end{equation}
The $100$ primitive repeats the current matching without moving atoms. It adds one pulse and preserves the entire placement.

The $111$ primitive rematches the two rightmost occupied sites. For pairs $(a,b)$ and $(c,d)$, a right-directed instance produces $(c,a)$ and $(b,d)$. Three atoms move with preserved relative order, while one remains fixed. Reflected instances alternate the direction to keep the workspace bounded. A changed partner requires a change of geometry; one explicit movement attains the resulting local lower bound.

The $223$ primitive allows $M$ to exceed $T$ by adding two routing intervals with three movement layers. Labels $0,\ldots,7$ below denote template roles assigned to the current input pairs. The primitive executes two new matchings:
\begin{align}
L_0&=\{(0,2),(1,3),(4,6),(5,7)\},\\
L_1&=\{(0,4),(1,5),(2,6),(3,7)\}.
\end{align}
One movement realizes $L_0$, followed by two movements that realize $L_1$. The final occupied site set equals the input site set, allowing the interface to be reused.

The movement bound is encoded as a satisfiability modulo theories (SMT) problem and checked with Z3~\cite{demoura2008}. Integer site coordinates, binary slot assignments, occupancy constraints, and order-preserving simultaneous motion define the state space. Initial pair sites and their slot orientations are free variables; the two complete output matchings execute in order. With at most two movement layers, the unbounded-coordinate formula is unsatisfiable. A three-movement witness is satisfiable with a reused final site set. These outcomes establish the local movement lower bound and its attainment.

\begin{table}[t]
\caption{Primitive contributions and local certificates. SAT and UNSAT denote satisfiable and unsatisfiable formulas.}
\label{tab:primitives}\centering\small
\begin{tabular}{@{}lccc l@{}}\toprule
Primitive & $\Delta R$ & $\Delta T$ & $\Delta M$ & Attainment\\\midrule
$100$ & 1 & 0 & 0 & Unchanged placement\\
$111$ & 1 & 1 & 1 & Order-preserving rematching\\
$223$ & 2 & 2 & 3 & $K=2$: UNSAT; $K=3$: SAT\\\bottomrule
\end{tabular}
\end{table}

\subsection{Composition and Resource Guarantees}
To realize a requested tuple, NAKO determines the number of instances of each primitive. Let $n_{100}$, $n_{111}$, and $n_{223}$ denote these counts. After the initial pulse, the constructor applies all $111$ instances, then $223$, and finally $100$. The $111$ prefix keeps its partner changes on the same four qubits, as used in the routing lower bound below; the $100$ suffix adds pulses without further routing. Including the initial pulse's contribution $(1,0,0)$ gives
\begin{align}
R_{\rm ref}&=1+n_{100}+n_{111}+2n_{223},\\
T_{\rm ref}&=n_{111}+2n_{223},\\
M_{\rm ref}&=n_{111}+3n_{223}.
\end{align}
Solving these equations gives the unique multiplicities
\begin{align}
n_{223}&=M-T,& n_{111}&=3T-2M,\nonumber\\
n_{100}&=R-1-T.&&
\label{eq:composition}
\end{align}
The feasible integer targets obey
\begin{equation}
T\leq M\leq 3T/2,\qquad R\geq T+1.
\end{equation}
For example, target $(16,12,18)$ consists of the initial matching, three $100$ primitives, and six $223$ primitives. Workload completion adds static CZ pairs to these same pulses without moving them, while single-qubit gates preserve the backbone CZ order. Section~\ref{sec:filler} details this extension.

\begin{theorem}[Constructive resource optimality]
For a feasible target and accepted grid, the generated schedule attains \rtm. Among dependency-preserving executions, it minimizes the Rydberg-stage count and routing-interval count. Among executions with the minimum Rydberg-stage count, it also minimizes the movement-layer count.
\end{theorem}
\begin{proof}
Each backbone qubit participates in $R$ ordered CZ gates, requiring at least $R$ pulses. Choose a qubit from the four rematched throughout the $111$ prefix. It changes partner once per $111$ primitive and twice per $223$ primitive. Each change requires movement between successive gates on that qubit. These gate-to-gate intervals are disjoint, giving $T\geq n_{111}+2n_{223}$.

In an execution using exactly $R$ pulses, every backbone qubit must participate in every pulse. Its $j$th gate therefore executes in pulse $j$, synchronizing each complete backbone matching. Every $111$ transition then requires at least one movement, and each $223$ window requires at least three by the local SMT bound, which permits arbitrary initial pair sites and orientations. The windows share only boundary pulses, so their movement costs add: $M\geq n_{111}+3n_{223}$. The constructed witnesses attain all these bounds; substitution of~\eqref{eq:composition} yields the target tuple. Static filler pairs use the same pulses, and single-qubit gates preserve the backbone CZ order.
\end{proof}

Alongside the logical gates, the construction records the initial placement and the movements between gate stages. These records form the executable reference schedule used for resource comparison. Its mappings and execution information also support the controlled tests of compiler decisions illustrated in Section~\ref{sec:diagnosis}.

\subsection{Filler and Physical Consistency}
\label{sec:filler}
Filler gates complete the requested workload after the backbone fixes the resource targets. The requested gate counts are $\lceil d_1QD\rceil$ and $\lceil d_2QD/2\rceil$, of which the backbone supplies $4R$ CZ gates. To adjust the two-qubit load without additional movement, NAKO places selected non-backbone qubits in static, disjoint pairs. Under the global-pulse model, each such pair interacts at every existing Rydberg stage and therefore contributes $R$ CZ gates. With $p$ selected pairs, the achieved count is
\begin{equation}
G_{2,\rm actual}=R(4+p).
\end{equation}
Because each static pair contributes $R$ gates, arbitrary two-qubit requests may not be exactly reachable. The generator chooses a feasible integer $p$ whose contribution is close to the remaining CZ budget. Single-qubit gates extend logical depth and fill available dependency positions. The one-pass procedure records achieved counts and any deviation from the requests. Static filler introduces no additional routing or movement: its atoms remain fixed, and its CZ gates reuse the existing pulses.

Export checks compare the declared gate pairs with the pairs induced by each reference pulse, check gate coverage and qubit order, and recompute \refcost. The generation record retains requested and achieved counts. This compact validation connects the logical workload, its physical witness, and the subsequent measurements.

\subsection{From Resource Evaluation to Controlled Diagnosis}
\label{sec:evaluation-diagnosis}
NAKO pairs each generated circuit with attainable resource targets and an executable reference schedule. Building on these outputs, we design a reference-guided workflow for compiler evaluation and diagnosis, comprising resource comparison, controlled intervention, and execution validation.

Resource comparison measures the tested compiler's $R$, $T$, and $M$ against the targets for the same circuit. The resulting profile identifies the resources with excess cost. Inspection of the compiled schedule then guides the selection of a decision to investigate, such as initial placement, cross-stage layout reuse, or gate-batch selection.

The intervention stage compares alternatives to the selected decision under matched conditions. To examine initial placement, the reference layout supplies an alternative mapping. Recompiling the same circuit with unchanged hardware and subsequent compiler settings measures the effect of this substitution. To examine a later scheduling or routing choice, we retain the executed operations and fix the compiler's complete physical state, including atom positions and trap assignments. Alternative gate and movement sequences are then constructed from this common starting point.

In the validation stage, the alternative executions are checked through the tested compiler's instruction generator and physical replay, or by fixing candidate assignments in its constraint model. Checks cover gate identities and per-qubit order, physical feasibility, and resource costs. A validated cost reduction establishes an improvement opportunity under the stated controls. Section~\ref{sec:diagnosis} applies this workflow to Enola and \DPQA.

\section{Experimental Methodology}
The experiments assess NAKO's ability to expose compiler resource overhead and support controlled diagnosis, together with the time and memory required to generate larger benchmarks.

\subsection{Benchmarks and Evaluation Tasks}
All benchmark circuits are generated by NAKO with their resource targets and executable reference schedules. Experiments A, B, and C examine compiler responses to circuit depth, qubit count, and CZ load. Experiments D and E examine primitive compositions and feasible resource sweeps, while F supplies four eight-qubit constructions for investigating individual compilation decisions. CZ-load comparisons use achieved gate counts.

Experiment G examines the diagnostic value of the reference initial layout by compiling the same five circuits under default initialization and the NAKO mapping. Experiment H separately measures the scalability of NAKO's benchmark generation. Experiment I compiles five fixed circuits across four hardware grids to examine resource use and completion under different capacities.

Table~\ref{tab:campaigns} summarizes the experimental design. Compiler experiments use five runs per parameter setting, except F, which contains one circuit per construction. H uses three repetitions per width. The design gives 324 compilation tasks per compiler.

\begin{table}[t]
\caption{NAKO-generated benchmark families and evaluation tasks. Run counts are per compiler except for H.}
\label{tab:campaigns}\centering\small
\begin{tabular}{@{}c>{\raggedright\arraybackslash}p{.60\columnwidth}rr@{}}\toprule
ID & Experimental factor & Points & Runs\\\midrule
A & Circuit depth: $D=10,20,30,40,50$ & 5 & 25\\
B & Qubit count: $Q=18,20,30,40,50$ & 5 & 25\\
C & Achieved CZ count:\newline 64, 80, 96, 112, 128, 144 & 6 & 30\\
D & Target $R$, $T$, or $M$ & 18 & 90\\
E & Exploratory resource sweeps & 24 & 120\\
F & Small primitive compositions & 4 & 4\\
G & Default and reference initial mappings & 2 & 10\\
H & Generator width:\newline $Q=100,500,1000,2000$ & 4 & 12$^{\dagger}$\\
I & Logical grid:\newline $2\times8$, $2\times10$, $2\times12$, $2\times16$ & 4 & 20\\\bottomrule
\multicolumn{4}{l}{$^{\dagger}$H measures NAKO generation.}
\end{tabular}
\end{table}

\subsection{Compilers and Experimental Settings}
The evaluation uses Enola~\cite{tan2025enola} and \DPQA~\cite{tan2024olsq}, two neutral-atom compilers with different compilation approaches. Enola separates gate scheduling, atom placement, and routing in a heuristic pipeline. \DPQA{} uses SMT-based layout and movement synthesis. These approaches provide complementary settings for examining the resource gaps and compiler decisions that NAKO can reveal.

Compiler experiments use the CZ portion of each generated circuit, retaining all qubits and the gate order on each qubit. This preserves the dependency condition used in NAKO's resource guarantees. Enola uses routing based on maximal independent sets and dynamic layout updates between Rydberg stages; \DPQA{} uses its greedy compilation mode with gate commutation disabled. In G, only the initial mapping changes between the paired runs; subsequent compiler settings remain fixed.

Outside I, both compilers use $16\times16$ spatial grids with 16 AOD rows and columns. In I, an $r\times c$ logical site grid is represented by a $2c\times r$ spatial grid in Enola and a $c\times r$ grid in \DPQA, with corresponding AOD dimensions. These configurations follow the compilers' respective representations of the two-slot site model.

Each compilation has a time limit of 10\,000~s, with up to eight compilations running concurrently per compiler. Compiler experiments run on Linux with Python 3.10.19; \DPQA{} uses Z3 5.1.0.

Generation scalability is measured separately on Windows with Python 3.12.14 and a fixed $64\times64$ site grid. Timing includes circuit construction, serialization, and QASM reload and statistics collection. Memory measurements report peak traced Python allocations.

\section{Evaluation Results}
The experiments compare compiler resource costs with NAKO's attainable targets and examine how the resulting gaps depend on workload size, interaction structure, initial placement, and hardware capacity. These comparisons guide the controlled investigations in Section~\ref{sec:diagnosis}. Generation time and memory are measured separately to assess the cost of producing larger benchmarks.
\subsection{Compilation Outcomes and Resource Costs}
Enola completes all 324 planned runs. OLSQ-DPQA completes 312; its 12 timeouts occur in I. The main resource summary covers experiments A--F, with 294 runs per compiler, and reports arithmetic means of per-instance resource ratios. Enola attains $R=R_{\rm ref}$ in every case. Its mean $T/T_{\rm ref}$ is 1.298 and $M/M_{\rm ref}$ is 36.914. OLSQ-DPQA has mean ratios 2.366, 2.984, and 2.073. Tables~\ref{tab:enola-summary} and~\ref{tab:dpqa-summary} preserve the separate resource profiles.

\begin{table}[t]\caption{Enola resource evaluation over 294 main entries.}\label{tab:enola-summary}\centering\small
\begin{tabular}{@{}lrrr@{}}\toprule
Resource & Mean ratio & Minimum ratio & Reference hits\\\midrule
$R$ & 1.000 & 1.000 & 294/294\\
$T$ & 1.298 & 1.000 & 34/294\\
$M$ & 36.914 & 15.571 & 0/294\\\bottomrule
\end{tabular}\end{table}
\begin{table}[t]\caption{OLSQ-DPQA resource evaluation over 294 main entries.}\label{tab:dpqa-summary}\centering\small
\begin{tabular}{@{}lrrr@{}}\toprule
Resource & Mean ratio & Minimum ratio & Reference hits\\\midrule
$R$ & 2.366 & 1.125 & 0/294\\
$T$ & 2.984 & 1.333 & 0/294\\
$M$ & 2.073 & 1.333 & 0/294\\\bottomrule
\end{tabular}\end{table}

These profiles identify different optimization priorities. Enola consistently packs the logical program into the reference number of entangling stages, while its expanded movement stream remains large. OLSQ-DPQA uses fewer movement layers in every jointly successful paired run, with additional entangling stages and routing intervals. The reference tuple allows these outcomes to be interpreted separately and directs the diagnostic probes in Section~\ref{sec:diagnosis}.

\subsection{Depth and Width Scaling}
Experiment A tests whether normalized resource gaps persist as the workload becomes deeper. At $Q=20$, logical depth increases from 10 to 50 while the reference tuple scales proportionally from $(8,6,9)$ to $(40,30,45)$. Figure~\ref{fig:expA} shows that Enola preserves the reference $R$ across the sweep; OLSQ-DPQA's normalized Rydberg and routing costs remain above two at most points. Absolute movement rises from 215.2 to 1110.0 for Enola and from 18.0 to 89.6 for OLSQ-DPQA. The persistent relative profiles indicate that longer workloads accumulate the same resource tradeoffs.

\figwide{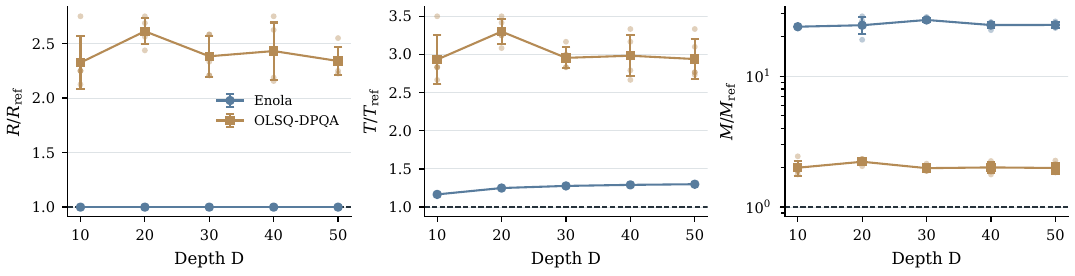}{Experiment A: depth scaling. From left to right, Rydberg, routing-interval, and movement ratios use the constructed reference as denominator. Points show individual runs, connected markers show means, and error bars show sample standard deviations. The dashed line marks reference attainment.}{fig:expA}

Experiment B tests how compiler costs change with register size when a reference execution with the same resource costs remains available. It fixes $D=20$ and reference $(16,12,18)$ while increasing $Q$ from 18 to 50. Figure~\ref{fig:expB} shows Enola's mean $M$ increasing from 349.0 to 972.8. OLSQ-DPQA's $M$ stays between approximately 35 and 43, while recorded compilation time rises from 30.54 to 581.70 s at the endpoints. This workload family exposes a growing movement burden in one implementation and a growing search burden in the other. Achieved density varies from 0.352 to 0.400 because static filler pairs contribute gates in units of $R$.

\figwide{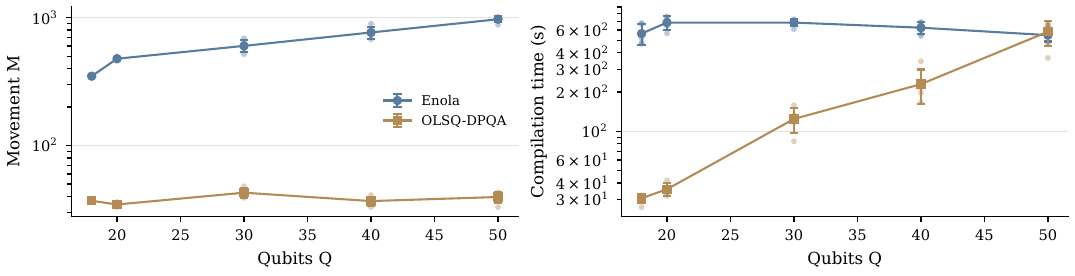}{Experiment B: width scaling at fixed reference $(16,12,18)$. Movement and recorded compilation time are shown separately on logarithmic vertical axes. Samples at $Q=18$ share the same two-qubit program.}{fig:expB}

\subsection{Interaction Load and Primitive Composition}
Experiment C tests the cost of additional interactions at fixed circuit width, depth, and reference resources. It varies achieved CZ count from 64 to 144 at $Q=20$, $D=30$, and reference $(16,12,18)$. Figure~\ref{fig:expC} shows that Enola retains $R/T=16/15$ while mean movement increases from 366.0 to 822.8. OLSQ-DPQA's mean tuple changes from $(37.0,33.0,33.0)$ to $(48.4,43.4,44.2)$. The controlled backbone keeps the reference fixed while additional static pairs expose how each compiler handles the surrounding interaction load.

\figwide{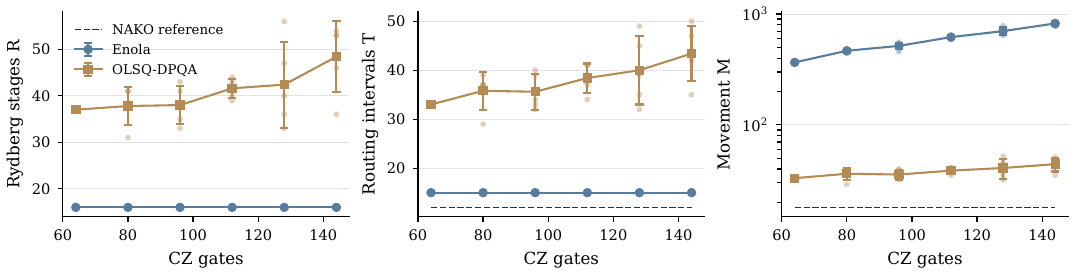}{Experiment C: achieved interaction load. Horizontal coordinates are the actual CZ counts; the corresponding densities span 0.2133--0.4800. Dashed curves give the fixed reference. Five runs are shown per point; these use one distinct two-qubit program at 64 CZ gates, four at 80, and five at the remaining points.}{fig:expC}

Experiment D tests sensitivity to primitive composition at fixed logical workload size. All circuits use $Q=20$ and $D=30$; the $M$ and $T$ branches also retain 128 CZ gates. In the $M$ branch, reference $R=16$ and $T=12$ are fixed. Increasing reference $M$ from 12 to 18 replaces $111$ with $223$ primitives; OLSQ-DPQA's mean tuple increases from $(22.4,20.8,20.8)$ to $(38.8,36.2,36.6)$. In the $T$ branch, reference $R=16$ and $M=18$ are fixed. Increasing reference $T$ from 12 to 15 reduces $n_{223}$ from six to three and increases $n_{111}$ from zero to nine. The mean OLSQ-DPQA tuple falls from $(40.0,38.2,38.8)$ to $(29.0,28.0,28.0)$.

\figwide{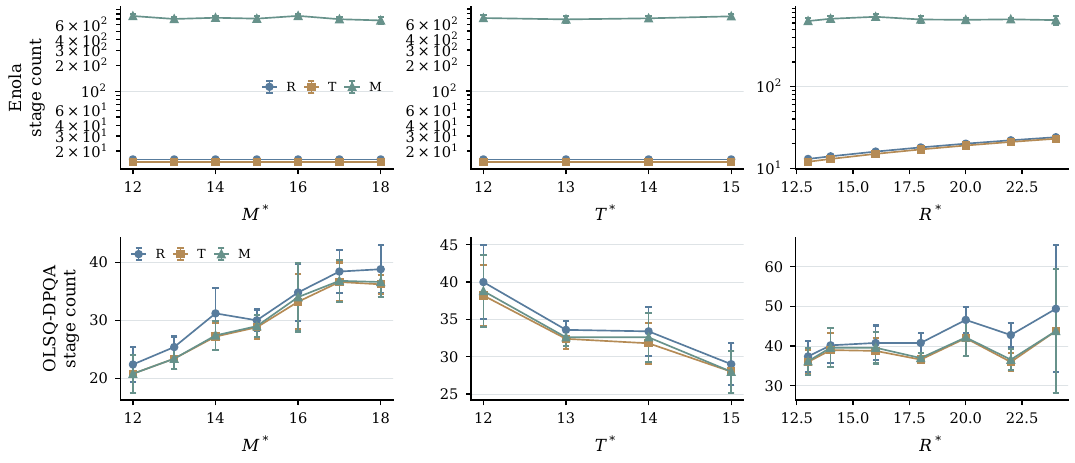}{Experiment D: separate sweeps of target $M$, $T$, and $R$ (left to right). The top row shows Enola and the bottom row OLSQ-DPQA. Marker shapes and colors identify the measured resource. Enola uses logarithmic vertical axes to show its Rydberg/routing counts and expanded movement count together. The $M$ and $T$ sweeps have 128 CZ gates at every point; the $R$ sweep has 110--128.}{fig:expI}

In the $T$ branch, these results correspond to mean $R/R_{\rm ref}$ decreasing from 2.50 to 1.81 and $M/M_{\rm ref}$ from 2.16 to 1.56. The reductions are measured against unchanged attainable $R$ and $M$ costs, with no decrease in qubit count, depth, or CZ count. They therefore reveal differences in compiler overhead across primitive compositions at the same workload size. The $R$ branch adds $100$ stages at fixed reference $T$ and $M$, while achieved CZ counts range from 110 to 128.

\subsection{Exploratory Resource Sweeps}
Experiment E examines compiler behavior under joint variations of resource targets and gate load. It uses two requested density vectors, $(d_1,d_2)=(0.27,0.36)$ and $(0.51,0.40)$. Each density uses two branches with six parameter points and five samples, giving 120 inputs per compiler. The movement branch uses $(M,T)=(12,8),(14,10),(16,11),(18,12),(20,14),(22,15)$ at $R=16$. The Rydberg branch holds $T=12$ and $M=18$ while varying $R=13,14,16,18,20,22$.

Figure~\ref{fig:expG} reports the resulting movement costs against the reference at both densities. The movement branch varies $T$ jointly with $M$, and filler quantization changes achieved gate counts in the Rydberg branch. These curves show the response to feasible combinations of resource targets and interaction load. Experiment D's fixed-load branches provide the more focused composition controls for interpreting that response.

\figwide{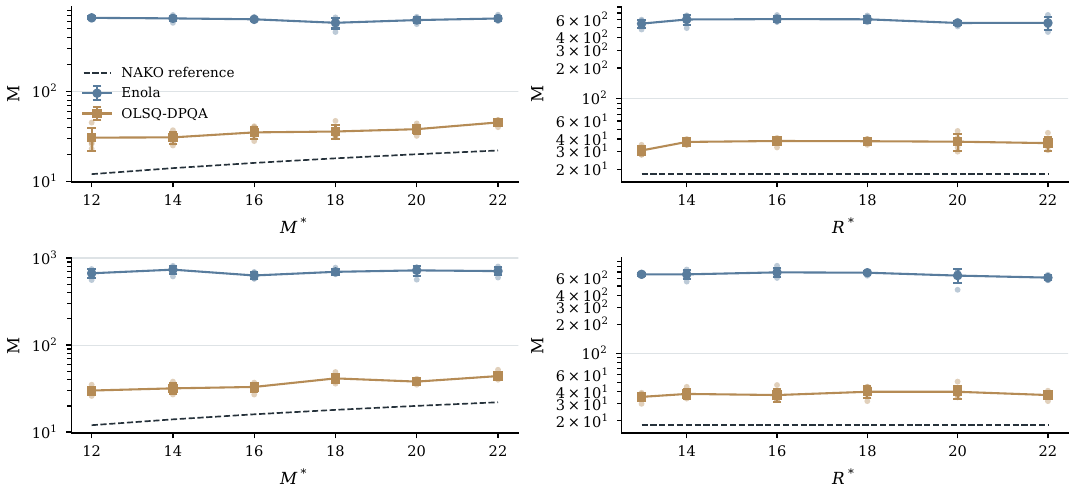}{Experiment E: movement counts under exploratory target sweeps. Rows correspond to requested densities $(0.27,0.36)$ and $(0.51,0.40)$; columns correspond to the movement and Rydberg branches. Dashed lines indicate the constructed reference. Movement sweeps change routing intervals jointly; Rydberg sweeps change achieved gate counts through filler quantization.}{fig:expG}

\subsection{Small Primitive Compositions}
Experiment F tests whether resource gaps already arise in small, directly inspectable primitive constructions. Its four eight-qubit circuits contain only $111$ primitives, one $223$ primitive, a mixed composition, or consecutive $223$ primitives. Enola attains the reference $R$ in all four but uses 53--154 movement layers, compared with 3--9 in the references. \DPQA{} uses 5--16 Rydberg stages, compared with reference counts of 3--7. These circuits form the inputs to the \DPQA{} fixed-state investigations in Section~\ref{sec:diagnosis}.

\subsection{Reference Initial Mappings}
Experiment G examines how initial placement affects the measured resource gaps. Each of five circuits, with reference $(16,12,18)$, is compiled using its default mapping and the initial mapping from its NAKO reference. Hardware and subsequent compiler settings remain fixed, and each compiler generates its own subsequent schedule.

Figure~\ref{fig:expJ} shows Enola's mean movement falling from 758.8 to 553.0, a 27.12\% reduction, with all five pairs improving. Enola retains $R=16$ and $T=15$ in both modes. OLSQ-DPQA's mean tuple changes from $(44.0,41.4,41.8)$ to $(33.6,31.4,31.6)$, reductions of 23.64\%, 24.15\%, and 24.40\%. Four of its five pairs improve all three resources; one pair increases slightly.

Against the common reference, Enola's mean $M/M_{\rm ref}$ decreases from 42.16 to 30.72 without increasing $R$ or $T$. Initialization substantially changes the magnitude of its movement gap, while movement remains well above the reference in both modes. The comparison between compilers also persists: under both initializations, OLSQ-DPQA uses fewer movement layers than Enola but more Rydberg stages and routing intervals.

\figwide{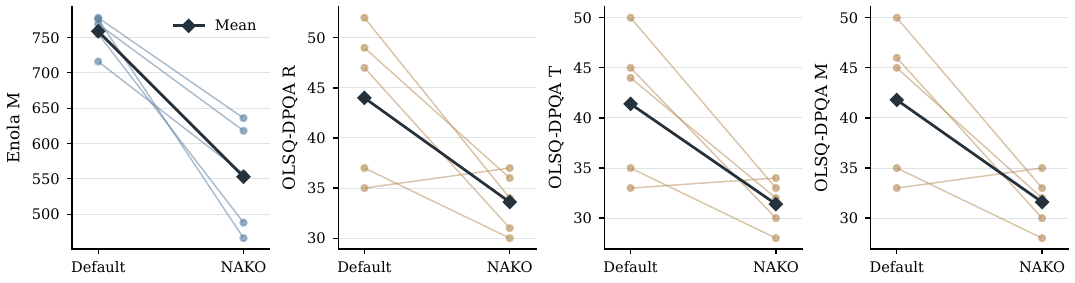}{Experiment G: paired default and NAKO mappings. The left plot shows Enola movement; the remaining plots show OLSQ-DPQA Rydberg stages, routing intervals, and movement. Thin lines connect each identical circuit across modes. Black diamonds give arithmetic means.}{fig:expJ}

\subsection{Generator Scalability}
Experiment H measures whether generation remains practical as qubit count increases from 100 to 2000. It uses a fixed $64\times64$ site grid and the same reference $(16,12,18)$ at all widths. Figure~\ref{fig:expK} shows mean generation times of 0.357, 1.354, 2.556, and 4.970 s for 100, 500, 1000, and 2000 qubits. Peak traced Python allocations increase from 0.51 to 6.22 mebibytes (MiB). The measurements include generation, serialization, and QASM reload/statistics, with three repeats per point. Direct witness construction makes large stress-test inputs practical without invoking a layout optimizer during generation.

\figsingle{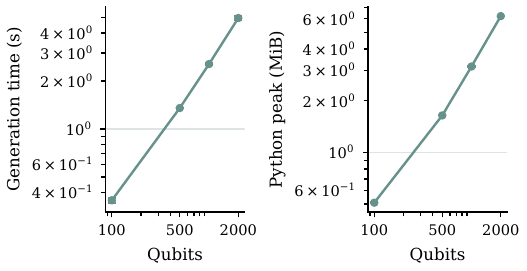}{Experiment H: generation time and peak Python allocation on a fixed grid. Error bars show the standard deviation over three repeats. Both axes use logarithmic scales.}{fig:expK}

\subsection{Hardware Capacity and Completion}
Experiment I tests how hardware space affects compilation quality and completion for an unchanged workload. It compiles the same five circuits on the four grids in Table~\ref{tab:campaigns}. Enola completes all 20 runs. OLSQ-DPQA completes 0/5 at $2\times8$, 0/5 at $2\times10$, 3/5 at $2\times12$, and 5/5 at $2\times16$. Figure~\ref{fig:expL} plots completion and the movement counts of successful runs. Enola's mean movement varies nonmonotonically with width. OLSQ-DPQA's 12 timeouts consistently end in repeated failed local queries, motivating the reliability branch in Section~\ref{sec:diagnosis}.

\figwide{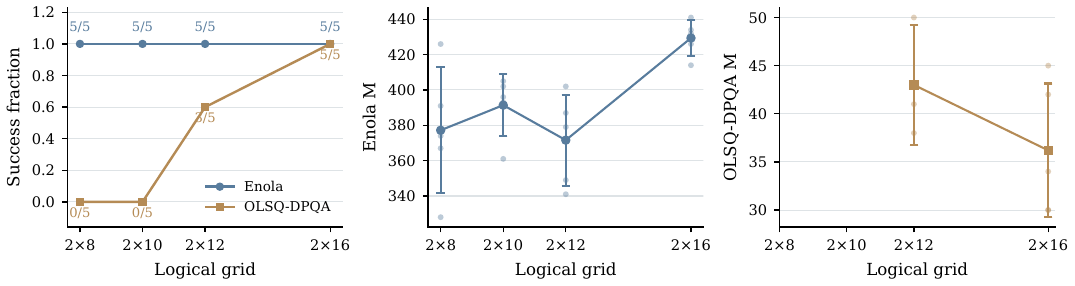}{Experiment I: completion fractions and movement on identical circuits across logical grids. The left plot labels successful runs out of five; movement plots contain successful runs only. OLSQ-DPQA has three observations at $2\times12$ and five at $2\times16$.}{fig:expL}

\section{Reference-Guided Compiler Diagnosis}
\label{sec:diagnosis}
The NAKO evaluation reveals different resource gaps in the two compilers. Enola attains the reference $R$ throughout the main workloads, while its mean movement ratio is approximately 37. \DPQA{} exceeds the reference $R$ and $T$ on the small primitive constructions and exhibits capacity-dependent completion on the hardware sweep. We trace these observations through controlled comparisons of compiler decisions, using reference layouts, small benchmark constructions, and validated alternative executions.

\subsection{Enola: Locating and Prioritizing Movement Losses}
The movement gap persists across NAKO's depth, width, and interaction-load experiments. Experiment G further shows that supplying the reference initial layout reduces Enola's movement while preserving its $R$ and $T$. Figure~\ref{fig:enola-diagnosis} follows this observed placement sensitivity through cross-stage layout controls, placement-score comparisons, and gathering-batch tests.

\figwide{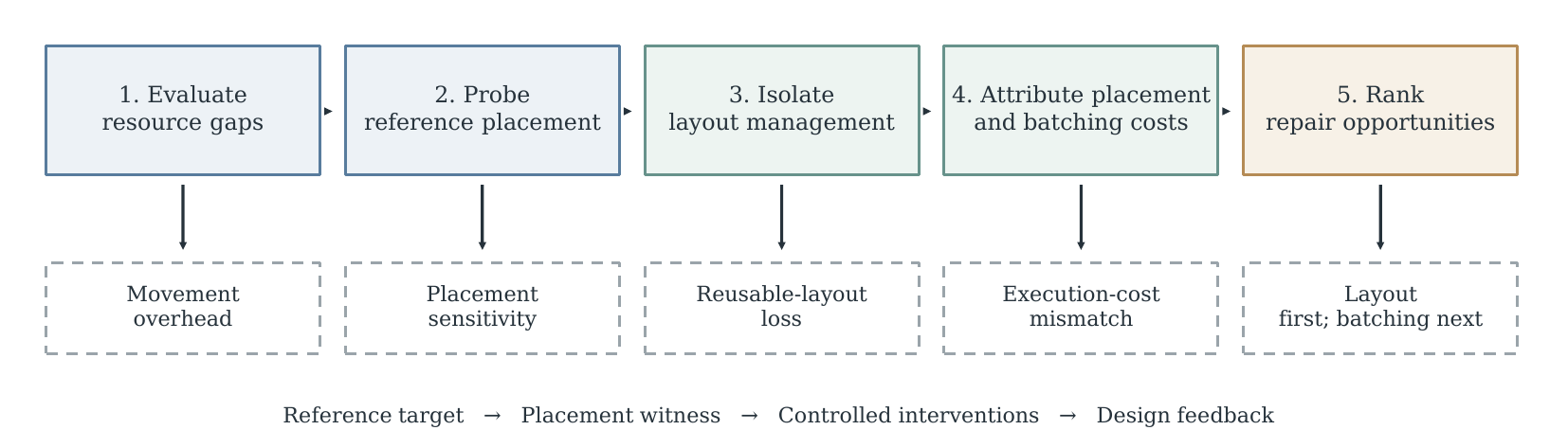}{NAKO's diagnostic workflow for Enola. Numbered process steps progress from resource evaluation through controlled placement and layout probes to attribution and repair priorities. The lower objects summarize findings produced by each step. Detailed interventions and measurements appear in the text and tables.}{fig:enola-diagnosis}

\subsubsection{Reference evaluation and placement probe}
The five G circuits each contain 128 CZ gates and have reference $(16,12,18)$. Replacing Enola's default initial layout with the NAKO layout reduces mean $M$ from 758.8 to 553, with $R=16$ and $T=15$ in both modes. Enola subsequently updates atom placement between stages and gathers gate pairs in batches of operations compatible with the AOD constraints. To examine how the supplied layout affects the full execution, we compare dynamic updates with repeated use of the initial configuration.

\subsubsection{Cross-stage layout management}
Enola's return-to-initial-layout option restores the initial atom positions after each Rydberg stage, giving successive stages the same starting layout. We apply both dynamic and return policies to the default and NAKO layouts. The scheduled CZ layers, hardware, router, and instruction generator remain fixed. Table~\ref{tab:enola-layout} reports the resulting costs.

\begin{table}[t]\caption{Enola layout and cross-stage policy controls on the same five G inputs. Fidelity is estimated using Enola's model.}\label{tab:enola-layout}\centering\small
\begin{tabular}{@{}llrr@{}}\toprule
Layout & Boundary policy & Mean $M$ & Fidelity\\\midrule
Default & Dynamic & 758.8 & 0.171\\
Default & Return & 777.2 & 0.167\\
NAKO & Dynamic & 553.0 & 0.188\\
NAKO & Return & 160.0 & 0.235\\\bottomrule
\end{tabular}\end{table}

With the NAKO layout, returning after each stage reduces mean $M$ from 553 to 160, approximately 71.1\%. With the default layout, the same policy increases $M$ from 758.8 to 777.2. The effect of preserving a configuration therefore depends on its layout. Repeated use of the NAKO configuration provides a substantially cheaper execution under Enola's own router, while dynamic updates increase the movement generated from that same initial layout. The default configuration remains expensive under both policies. These results identify a strong dependence between layout choice and the cost of cross-stage updates.

\subsubsection{Placement objective and execution cost}
The fixed-boundary executions produce a large cost difference between the two layout families: mean $M$ is 777.2 for the default layouts and 160 for the NAKO layouts. To examine how the compiler evaluates these configurations, we compute its placement objective on both. The objective minimizes a weighted sum of gate-endpoint distances over upcoming CZ layers. It assigns approximately 54.6 to the default layouts and 75.6 to the NAKO layouts, ranking the more expensive executions ahead of the cheaper ones.

This ranking reversal on the NAKO inputs leads to an independent check using a single four-CZ matching. One layout scores 9.37 and produces 16 movement layers with a modeled duration of about 1.00~ms; another scores 9.82 and produces 14 movement layers with a duration of about 0.75~ms. Each layout uses its minimum gathering-batch count. The lower-distance layout remains more expensive after batch-count optimization within each layout. Gate-endpoint distance evaluates pairs individually, while physical execution also depends on compatible gathering operations and their loading and unloading sequence. The two comparisons expose a concrete mismatch between the placement proxy's ordering and the costs of executable plans.

\subsubsection{Routing batches and repair priorities}
On the same five NAKO inputs with fixed stage boundaries, Enola's original selector uses averages of 34 gathering batches for the default layouts and 20 for the NAKO layouts. To measure the contribution of batch selection, we compare its maximal-independent-set heuristic with enumeration of the minimum batch count for each layout. Enumeration considers both gathering directions of each gate under the original compatibility test. Both selectors use the same router and instruction generator. Table~\ref{tab:enola-batch} summarizes the results.

\begin{table}[t]\caption{Enola batching controls with fixed stage boundaries. Minimum batches refers to the number of gathering batches.}\label{tab:enola-batch}\centering\small
\begin{tabular}{@{}llrr@{}}\toprule
Layout & Selector & Mean batches & Mean $M$\\\midrule
Default & Original & 34.0 & 777.2\\
Default & Minimum batches & 32.8 & 754.4\\
NAKO & Original & 20.0 & 160.0\\
NAKO & Minimum batches & 20.0 & 160.0\\\bottomrule
\end{tabular}\end{table}

For the default layouts, enumeration reduces mean batch count from 34 to 32.8 and mean $M$ from 777.2 to 754.4. Mean modeled fidelity improves by approximately 1.8\%; the strongest sample removes three batches and reduces movement by 8.2\%. For the NAKO layouts, both selectors produce 20 batches and $M=160$. Batch selection therefore provides additional savings for the default configurations, while the preserved NAKO configurations already attain the minimum batch count.

The controlled executions preserve the scheduled CZ layers and pass physical replay of trap activation, coordinate updates, global-pulse gate pairs, and final states. Additional quantum-volume (QV) and QFT matching-slice tests use 40 preset layouts each. Enumeration removes excess gathering batches in nine QV layouts and improves movement and modeled fidelity in each. The original selector attains the minimum batch count in all 40 tested QFT layouts.

The movement gap measured against NAKO's targets is traced to placement-score mismatch and costly updates of the reference layout. Reusing that layout reduces mean $M$ from 553 to 160 through Enola's own router and instruction generator, while minimum-batch selection provides smaller, layout-dependent gains (Tables~\ref{tab:enola-layout} and~\ref{tab:enola-batch}). NAKO supplies both the attainable resource target that directs the investigation and the executable layout used to test compiler choices. These controls yield validated movement reductions and establish layout selection and cross-stage reuse as the primary improvement priorities for the five inputs.

\subsection{OLSQ-DPQA: Connecting Quality and Reliability}
NAKO's four small primitive constructions exhibit excess Rydberg stages and routing intervals even at eight qubits. The hardware-capacity experiment additionally shows declining completion on tighter grids for unchanged circuits. Figure~\ref{fig:dpqa-diagnosis} develops these observations into two investigations: the cost of accepted gate batches and the progression of local search after failure.

\figsingle{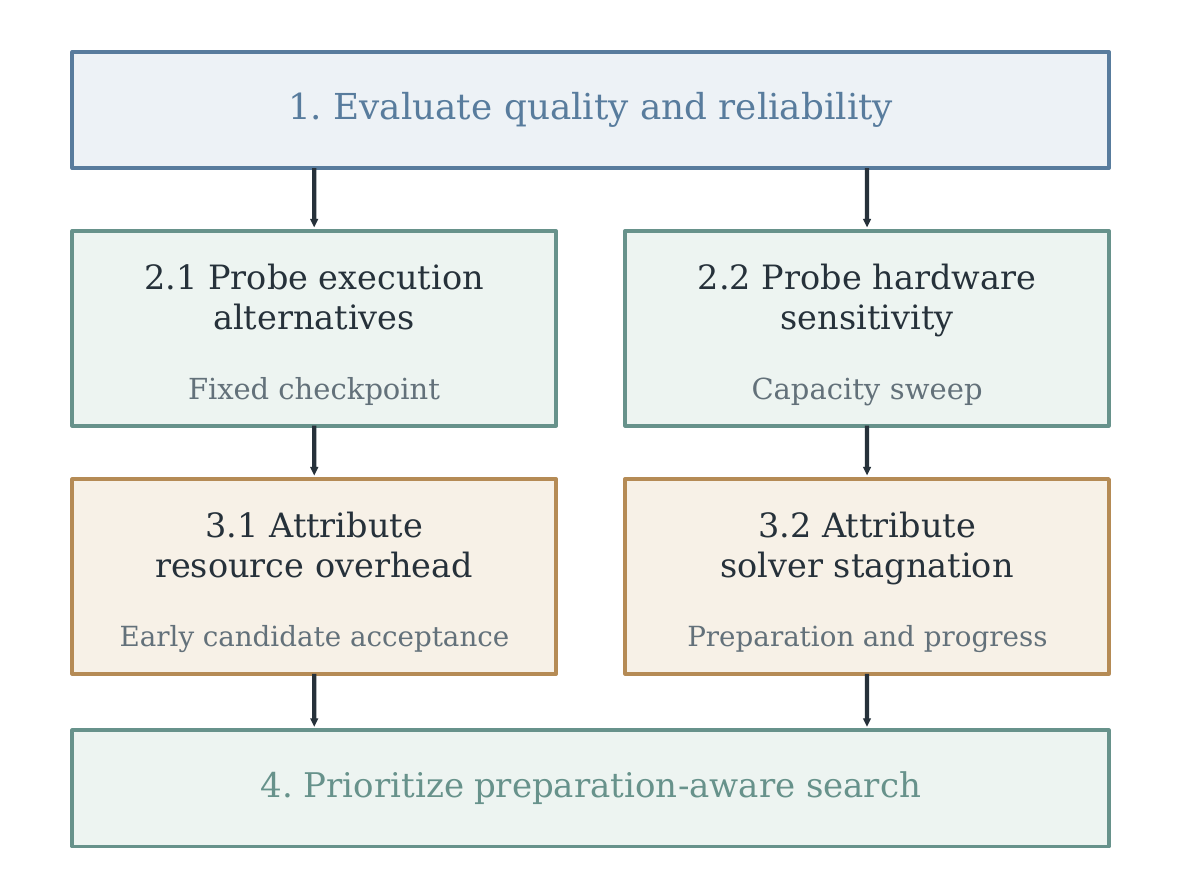}{NAKO's diagnostic workflow for OLSQ-DPQA. A common evaluation branches into fixed-checkpoint execution probes and hardware-capacity sweeps. The branches attribute resource overhead and solver stagnation, then converge on preparation-aware search as a design priority. Step numbers indicate the diagnostic sequence.}{fig:dpqa-diagnosis}

\subsubsection{Quality branch: fixed-state alternatives}
The one-$223$ input in experiment F contains eight qubits and twelve CZ gates arranged in three distinct successive matchings. Its NAKO reference is $(3,2,3)$, while \DPQA{} produces $(5,4,4)$. To examine where the additional pulses arise, we trace its gate-batch decisions and select the state after the third pulse, where four CZ gates remain ready for execution.

We retain the complete executed prefix and fix atom coordinates, trap assignments, and AOD row and column identifiers at this checkpoint. These variables preserve the physical starting conditions and shared-motion groups for the subsequent operations. \DPQA{} selects a batch through a local constraint query whose horizon specifies the number of modeled state transitions.

From this state, one-step queries are unsatisfiable when required to execute four or three gates and satisfiable for two gates. The controller accepts the two-gate batch and executes the remaining two gates after another movement. Its suffix places a pulse after each of two movement layers. A candidate generated under the same physical constraints uses the first movement to prepare the geometry and the second to bring all four gate pairs together for one pulse. Both suffixes use two movement layers, while the preparatory candidate combines them into one routing interval before a single Rydberg stage. The complete program consequently changes from $(5,4,4)$ to $(4,3,4)$.

The resource reduction is obtained with the circuit, prefix, complete physical state, hardware, and movement budget fixed. Inspection of the controller explains the observed choice: it decreases the required gate count within the one-step query and accepts the first feasible bound. It advances to a two-step horizon after every positive one-step bound fails. At this checkpoint, the feasible two-gate batch triggers acceptance before the longer preparation candidate enters the search.

\begin{table}[t]\caption{OLSQ-DPQA fixed-checkpoint continuations. Each alternative saves one Rydberg stage and one routing interval at unchanged movement.}\label{tab:dpqa-quality}\centering\small
\begin{tabular}{@{}lccc@{}}\toprule
Construction & Reference & Original & Alternative\\\midrule
All-$111$ & $(4,3,3)$ & $(6,5,5)$ & $(5,4,5)$\\
One-$223$ & $(3,2,3)$ & $(5,4,4)$ & $(4,3,4)$\\
Mixed & $(6,5,7)$ & $(16,15,15)$ & $(15,14,15)$\\
Consecutive-$223$ & $(7,6,9)$ & $(16,15,15)$ & $(15,14,15)$\\\bottomrule
\end{tabular}\end{table}

Table~\ref{tab:dpqa-quality} repeats the comparison across all four F constructions. Each alternative saves one Rydberg stage and one routing interval at unchanged movement. The complete candidate programs preserve gate identities, per-qubit order, and the gate-pair relation induced at each global pulse. Fixing candidate assignments in the compiler's constraint model confirms feasibility. The repeated result connects the observed $R/T$ gaps to a specific acceptance rule: immediate partial execution receives priority over preparation that enables a larger joint execution.

\subsubsection{Reliability branch: preparation and progress}
Experiment I compiles the same five NAKO circuits across progressively tighter grids. \DPQA{} completes all five on the widest grid and times out on all five on the two tightest grids. Examining the 12 timeout traces reveals thousands of repeated local failures after the last compilation progress. The controller repeatedly returns to a two-step query. Its source resets the horizon to two whenever the gate bound reaches zero, reproducing that search horizon at the unchanged execution state.

The dependence on hardware capacity directs the analysis to the query's site-occupancy constraints. Its interaction-exactness condition applies at every modeled state after the initial checkpoint: co-located atom pairs must execute gates at that state. For a matching of ready gates, $G_s$ executed gates pair $2G_s$ atoms, while the remaining atoms occupy individual sites. A state with $N$ atoms therefore occupies at least $N-G_s$ sites. With $S$ available sites,
\begin{equation}
G_s\geq\max(0,N-S).
\label{eq:capacity}
\end{equation}
Across $k$ encoded states, each ready gate is executed once, so the query requires at least $k\max(0,N-S)$ ready gates. At the $2\times8$ checkpoint of sample I-4, 20 atoms share 16 sites and six gates are ready. Each encoded state requires at least four gates, giving a requirement of eight across the two-state query. This exceeds the available six gates and explains the query's unsatisfiability at that checkpoint.

To examine state representation and search horizon separately, we construct a four-qubit, one-CZ control with a fixed initial physical state. The queries vary hardware width, horizon, and the condition under which interaction exactness applies. In the pulse-guarded model, exactness applies when a Rydberg pulse is active; movement-only states can retain co-located atoms while preparing subsequent gates. The AOD motion, transfer, ordering, and crowding constraints are retained. Table~\ref{tab:dpqa-reliability} reports the query outcomes.

\begin{table}[t]\caption{OLSQ-DPQA preparation-state controls. Each triple gives the results at horizons one, two, and three. U denotes unsatisfiable and S denotes satisfiable.}\label{tab:dpqa-reliability}\centering\small
\begin{tabular}{@{}ccc@{}}\toprule
Width & Original exactness & Pulse-guarded exactness\\\midrule
3 & U / U / U & U / U / S\\
4 & U / S / S & U / S / S\\
5 & U / S / S & U / S / S\\\bottomrule
\end{tabular}\end{table}

At width three, the original model is unsatisfiable at horizons one, two, and three. With pulse-guarded exactness, horizon three produces a feasible execution, while horizons one and two remain unsatisfiable. Fixed-assignment replay confirms the three-step witness. Preparation in this control therefore depends on both the admitted intermediate states and the explored horizon. Tracing the original controller with the pulse guard enabled records a failed one-step query followed by repeated identical two-step queries. The feasible three-step execution identifies the longer horizon needed to progress from that same state.

\subsubsection{A unified design direction}
The quality and reliability results identify three linked decisions: accepting a gate batch, representing preparation states, and advancing search after failure. The resulting design priorities are to compare immediate and preparatory candidates by incremental $R$, $T$, and $M$, apply interaction exactness under explicit pulse activation, and expand the horizon after a failed state--horizon query.

Across the four NAKO constructions, validated suffix replacements save one Rydberg stage and one routing interval at unchanged movement (Table~\ref{tab:dpqa-quality}). These improvements are obtained from the compiler's own checkpoints and locate a recoverable part of the $R/T$ gap in its batch-acceptance rule. The compact-hardware controls additionally identify preparation-state encoding and horizon progression as causes of local search failure (Table~\ref{tab:dpqa-reliability}). NAKO's attainable targets reveal the resource excess, and its small circuit families support the fixed-state comparisons that turn that observation into lower-cost executions and specific compiler design changes.

\section{Conclusion}
NAKO evaluates neutral-atom compilers through parameterized circuits with attainable R, T, and M targets and executable reference schedules. Under dependency-preserving execution, its backbone and composable primitives establish minimum R and T, with M minimal among schedules attaining minimum R. Comparisons with these references quantify compiler resource gaps, while the benchmark families support systematic evaluation of workload sensitivity. Direct construction also makes large-scale testing practical, producing 2000-qubit benchmarks in about five seconds.

The observed resource gaps guide controlled investigations of compiler decisions. Reference-layout comparisons expose placement-score mismatch and costly cross-stage updates in Enola; validated interventions reduce mean M by 78.9\% on five matched circuits. The OLSQ-DPQA study links extra stages to batch acceptance and local search failures to preparation-state representation and horizon progression. Four validated suffix replacements each save one Rydberg stage and one routing interval at unchanged movement. These cases show how NAKO's attainable targets and physical references connect measured resource overhead to specific compiler decisions and verified execution improvements.

\bibliographystyle{IEEEtran}
\bibliography{references_tcad}
\end{document}